\documentclass[10pt, conference, letter]{IEEEtran}

\usepackage[noadjust]{cite}

\usepackage{verbatim}

\usepackage{color}
\usepackage{graphicx}
\usepackage{amsthm}

\usepackage[cmex10]{amsmath}
\usepackage{amssymb}
\usepackage{mathtools}
\usepackage{float}
\usepackage{tikz}
\usepackage{epstopdf}
\usepackage{algorithmic}
\usepackage{algorithm}

\graphicspath{ {figures/} }

\newtheorem{theorem}{Theorem}
\newtheorem{lemma}{Lemma}

\newcommand{\saeedLater}[1]{}

\newcommand{\saeedDissertation}[1]{}

\newcommand{\doc}[1]{}

\usetikzlibrary{positioning}
\usetikzlibrary{fit}

\tikzstyle{axis-line}=[thick, ->, >=stealth]
\tikzstyle{axis-tick}=[thin]
\tikzstyle{arrow} = [thick,->]

\tikzstyle{block}=[rectangle,black,draw, text centered, minimum height=0.7cm]
\tikzstyle{summul} = [draw, shape=circle, text centered,inner sep=0]

\title{The Method of Conditional Expectations\\ for Cubic Metric Reduction in OFDM}
\date{\today}
\author{\IEEEauthorblockN{Saeed Afrasiabi-Gorgani, Gerhard Wunder}
\IEEEauthorblockA{Heisenberg Communications and Information Theory Group\\
Free University of Berlin\\
s.afrasiabi@fu-berlin.de, g.wunder@fu-berlin.de}}

\begin{document}

\maketitle

\begin{abstract}
High variations in the OFDM signal envelope  cause nonlinear distortion in  the power amplifier of the transmitter, which is a major drawback. Peak-to-Average Power Ratio (PAPR) and Cubic Metric (CM) are commonly used for quantifying this characteristic of the signal. Despite the reportedly higher  accuracy compared to PAPR, limited research has been done on reduction algorithms for CM.  In this paper, the Method of Conditional Expectations (CE Method) is used to achieve CM reduction by the Sign Selection approach.  Using the CE Method, the amenable mathematical structure of CM is exploited to develop a low complexity algorithm. In addition, guaranteed reduction is analytically proved for every combination of the data symbols.  Simulations show a reduction gain of almost 3~dB in Raw Cubic Metric (RCM) for practically all subcarrier numbers, which is achieved using only half the full rate loss of the Sign Selection approach. 
\end{abstract}


\section{Introduction}

Orthogonal Frequency Division Multiplexing (OFDM) is a well-known multicarrier waveform which has been used in a number of major wireless communication systems. However, a drawback of the OFDM signals is their high dynamic range, which causes nonlinear distortion at the output of the power amplifier. In addition to the performance degradation in the receiver, the nonlinear distortion creates Out-Of-Band (OOB) radiation and can violate the spectral mask. The problem particularly becomes a bottleneck in physical layer design in the mobile terminals due to the low-cost amplifiers and limited battery life.

The problem is commonly formulated as the minimization of a metric which captures the physical phenomenon. The classical metric is the ratio of the peak power to the average power  over consecutive signal segments and is referred to as Peak-to-Average-Power-Ratio (PAPR).   A more recently proposed metric, referred to as Cubic Metric (CM), is employed in the modern communication systems \cite{DahlmanLTE2014} and is reported to be more accurate than PAPR. Several proposals are available in the literature for the CM reduction problem. For instance, Clipping and Filtering is considered in \cite{ZhuDescClip2013}, which is based on clipping the signal in a controlled manner such that OOB radiation is limited. Some distortion-less methods are 
Partial Transmit Sequence (PTS) as used in \cite{ParkPTS2018} and Tone Reservation (TR)  as in \cite{Behravan2011}. In distortion-less methods, some resources are reserved to allow modification of the signal in order to reduce the desired metric. 

Using the sign of the symbols that modulate the subcarriers has been shown to be a promising method in PAPR reduction problem, referred to in this paper as Sign Selection  \cite{Sharif2004constantPMEPR, Afrasiabi2015derandomized,TellamburaGuidedSS2018,Sharif2009sign,tellamburaCrossEntropy2008}. There are $2^N$ possible sign combinations, with $N$ being the number of subcarriers, which makes the minimization problem of exponential complexity. This has motivated research on competing suboptimal solutions. Some proposals with noticeable performance include the application of the method of Conditional Probabilities in \cite{Sharif2004constantPMEPR, Afrasiabi2015derandomized}, a sign selection method guided by clipping noise in \cite{TellamburaGuidedSS2018}, a greedy algorithm in  \cite{Sharif2009sign} and a cross-entropy-based algorithm in \cite{tellamburaCrossEntropy2008}. The CM reduction problem, however, has not been tackled by the Sign Selection approach.

In this paper,  the method of Conditional Expectations (CE Method), which is originally proposal in graph theory \cite{MitzenmacherUpfal2005}, is used to obtain a suboptimal solution to the Sign Selection problem for CM reduction. Briefly, the CE Method introduces an artificial randomness to the signal and uses the conditional expectations to make the sequential sign selection possible. Asymptotic distribution of the modified signal is shown to lay the basis for derivation of simple closed-form expressions for a sequential sign selection rule. The CM reduction performance of the method is then analyzed to obtain an accessible upper bound on the worst-case reduced CM. Finally, simulations are used to verify the analysis and to show significant reduction in CM for a wide range of $N$. Furthermore, it is shown that only $\frac{N}{2}$ sign bits are enough to achieve almost the same reduction performance, which implies a considerably lower rate loss.

\paragraph*{Notation} A random variable $X$ is distinguished from a realization $x$ by using upper and lower case letters, respectively. Vectors are shown by bold-face letters. For a vector $\mathbf{x}$, the notation $x_{m:n}$ is the compact form for $[x_m, x_{m+1},\ldots, x_n]$. The expected value of $Y$ with respect to the random variable $X$ is denoted by $\mathbb{E}_X[Y]$, where the subscript may be omitted, if clear from the context. 

\section{Signal model and Cubic Metric}
\label{sec:pre}

Consider an OFDM scheme with $N$ subcarriers modulated by data symbols generated  independently and equiprobably from the set of constellation points $\mathcal{M}$. Let the random vector $\mathbf{B}\in \mathcal{M}^N$ denote the data symbols. 
The oversampled discrete-time baseband OFDM symbol  is
\begin{equation}
	s(n, \mathbf{B})\!=\!\frac{1}{\sigma_b \sqrt{N}}\sum_{k=0}^{N-1}\!\! B_k e^{j \frac{2\pi}{LN} k n} ,\ n=0,1,\ldots, LN\!-\!1,
	\label{eq:ofdmSymbolDef}
\end{equation}
where the signal is normalized by $\sigma_b^2=\mathbb{E}[|{B_k}|^2]$ in order to have unit power  and  $L>1$ is  the oversampling factor required for reliable calculation of CM \cite{KimCubicMetric2016}.

\textbf{Cubic Metric (CM)}  is based on the energy of the third-order signal which appears at the output of the power amplifier and is the main source of the nonlinear distortion \cite{CMmotorola}. For the OFDM signal $v(t)$, CM is defined as 
\begin{equation*}
	\mathrm{CM}_\mathrm{dB}=\frac{\mathrm{RCM_\mathrm{dB}}-\mathrm{RCM}_\mathrm{ref, dB}}{K_\mathrm{slp}}+K_\mathrm{bw},
\end{equation*}
where the Raw Cubic Metric (RCM) is
\begin{equation}
	\mathrm{RCM}_\mathrm{dB}=20 \log_{10}\left( \mathrm{rms}\!\left[\left(\frac{v(t)}{\mathrm{rms}[v(t)]}\right)^3 \right]\right)
	\label{eqn:RCMdef}
\end{equation}
and $\mathrm{RCM}_\mathrm{ref}$ is obtained similarly for a reference signal. The slope factor $K_\mathrm{slp}$ and the bandwidth scaling factor $K_{bw}$ are obtained from hardware measurements \cite{CMmotorola}. The Root Mean Square (RMS) of a signal, e.g. $\mathrm{rms}[v(t)]$, over a large interval $T\subset \mathbb{R}$ is $\sqrt{\frac{1}{T} \int_T v(t) dt}$.  Note that CM and RCM are scalars obtained from the whole signal. The reduction algorithms, on the other hand, typically operate on individual OFDM symbols. Therefore, RCM of an OFDM symbol is used which is defined below.
 
\textbf{Symbol RCM (SRCM)} for an oversampled discrete-time  baseband OFDM symbol is
\begin{align}
	\eta_{\scriptscriptstyle N}(\mathbf{B})=\frac{1}{LN} \sum_{n=0}^{LN-1} |s(n,\mathbf{B})|^6,
	\label{eqn:SRCM}
\end{align}
where $\mathrm{rms}[v(t)]=1$ is assumed. It is straight-forward to obtain the relation between $v(t)$ and $s(n)$. For instance, refer to \cite[Ch.~14]{Benedetto1999}.

\section{Sign Selection Problem}
\label{sec:signselection}

 In the Sign Selection problem, the sign bit of each data symbol is reserved to form the optimization variables $\mathbf{x}\in\{-1,1\}^N$.  Assume that the constellation is symmetric such that for each point $y\in\mathcal{M}$, the negated value $-y$ is in the set. Let $\mathcal{M'}\subset\mathcal{M}$ be a (non-unique) choice of  $|\mathcal{M}|/2$ points of $\mathcal{M}$ such that if $y\in\mathcal{M}'$,  then $-y\notin\mathcal{M}'$.  A sample choice of $\mathcal{M}'$ for 16-QAM is shown in Fig.~\ref{fig:choiceOfConst}. Then 
\begin{equation*}
	\mathcal{M}^N=\{\mathbf{c}\odot\mathbf{x}:\mathbf{c}\in\mathcal{M}'^N, \mathbf{x}\in\{-1,1\}^N\},
\end{equation*}
where $\odot$ denotes element-wise multiplication of vectors. For each $\mathbf{c}$, the Sign Selection problem suggests a solution $\mathbf{x}^\ast$. Consequently, $\mathbf{c}\odot\mathbf{x}^\ast$ will be the transmitted symbols.  
The optimization problem for a given $\mathbf{c}$ can be represented as
\begin{equation}
	\min_{\mathbf{x}\in\{-1,1\}^N} f(\mathbf{c}\odot\mathbf{x}),
	\label{eqn:optimizationProblem}
\end{equation}
where $f(.)\geq 0$ is the metric, such as PAPR or CM. 

\begin{figure}[t]
	\centering
	\includegraphics[width=0.5\columnwidth]{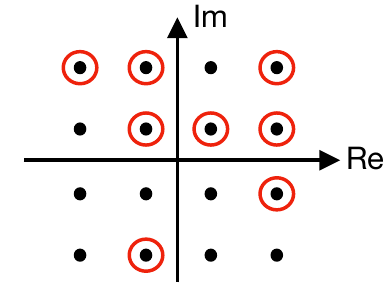}
	\caption{A non-unique choice of $\mathcal{M}'$ from $\mathcal{M}$; here a 16-QAM constellation.}
	\label{fig:choiceOfConst}
	\vspace{-0.4cm}
\end{figure}

\subsection{Receiver processing and Rate loss}

For the Sign Selection scheme, decoding is simply performed by choosing $b\in\mathcal{M}'$ when one of $\pm b\in\mathcal{M}$ is detected. That is, the decoding adds no complexity to the receiver. Consider that using $\mathcal{M}'$ for the symbols with a reserved bit for sign selection incurs a rate loss. If $N_s\leq N$ signs are used in the sign selection by mapping the corresponding data symbols from $\log_2 |\mathcal{M}|-1$ bits to $\mathcal{M}'$ and the remaining $N-N_s$ data symbols from $\log_2 |\mathcal{M}|$ bits to $\mathcal{M}$, the incurred amount of rate loss is 
\begin{align*}
	R= \frac{N_s \log_2 \frac{|\mathcal{M}|}{2}}{N \log_2 |\mathcal{M}|}
	= \frac{N_s}{N} \log_{|\mathcal{M}|} 2,
\end{align*}
which clearly decreases for a larger constellation.

\section{Method of Conditional Expectations}
\label{sec:CEmethod}

The CE Method \cite{spencer10}, \cite{MitzenmacherUpfal2005} is represented here for reduction of an arbitrary metric $f(.)\geq 0$ by sign selection. For a given data vector $\mathbf{c}$, the value of the sign variables will be decided sequentially. Consider the random vector $\mathbf{X}\in\{-1,1\}^N$ with independently and equiprobably distributed elements. At the $j$-th iteration where $x^\ast_{0:j-1}$  are already decided,  $x^\ast_j$ is chosen such that the expectation of $f(\mathbf{c}\odot \mathbf{X})$  conditioned on $X_{0:j-1}=x^\ast_{0:j-1}$ is minimized. Formally, the solution  is obtained by the decision rule
\begin{equation}
		x_j^\ast = \underset{x_j\in\{\pm1\}}{\mathrm{arg\ min}} \ \mathbb{E} [f(\mathbf{c}\odot\mathbf{X})|X_{0:j-1}=x^\ast_{0:j-1},X_j=x_j]
		\label{eqn:CEGeneralRule}
	\end{equation}
for $j=0,1,\ldots,N-1$. Justification of the fact that the above decision rule results in a $\mathbf{x}^\ast$ such that $f(\mathbf{c}\odot \mathbf{x}^\ast)$ is a desirably reduced metric value is partly available form the standard treatment of the method \cite{MitzenmacherUpfal2005} and is represented here for the general metric $f$. It will be elaborated for SRCM in Section~\ref{sec:analysis} and verified via simulation in Section~\ref{sec:perf}. 

 Let
\begin{equation}
	g_j^\pm(\mathbf{c})=\mathbb{E} [f(\mathbf{c}\odot\mathbf{X})|X_{0:j-1}=x^\ast_{0:j-1},X_j=\pm1].
	\label{eqn:gDefinition}
\end{equation}
After the $j$-th sign decision, 
\begin{align*}
	\mathbb{E} &[f(\mathbf{c}\odot\mathbf{X})|X_{0:j} =x^\ast_{0:j}]=\min\  \{g_j^+(\mathbf{c}), g_j^-(\mathbf{c})\}.
\end{align*}
In addition, it holds that
\begin{align*}
	\mathbb{E} &[f(\mathbf{c}\odot\mathbf{X}) | X_{0:j-1}=x^\ast_{0:j-1}] \nonumber \\
	& \quad\quad =g_j^+(\mathbf{c}) \mathbb{P}(X_j=1)  +  g_j^-(\mathbf{c}) \mathbb{P}(X_j=-1) \nonumber \\
	&\quad \quad =\frac{1}{2}(g_j^+(\mathbf{c}) + g_j^-(\mathbf{c})) \nonumber \\
	& \quad \quad \geq \min \{g_j^+(\mathbf{c}), g_j^-(\mathbf{c})\}.
\end{align*} 
Therefore, by each sign decision, we have
\begin{equation*}
	\mathbb{E}[f(\mathbf{c}\odot\!\mathbf{X})|X_{0:j}\!=\!x^\ast_{0:j}]\!\leq\! \mathbb{E}[f(\mathbf{c}\odot\!\mathbf{X})|X_{0:j-1}\!=\!x^\ast_{0:j-1}],
\end{equation*}
which shows a non-increasing sequence of conditional expectations. For a given $\mathbf{c}$, it begins with the \emph{initial expectation} $\mathbb{E}_\mathbf{X}[f(\mathbf{c}\odot\mathbf{X})]$ with complete randomness in $\mathbf{X}$. At last, it ends with $f(\mathbf{c}\odot\mathbf{x}^\ast)=\mathbb{E} [f(\mathbf{c}\odot\mathbf{X})|\mathbf{X}=\mathbf{x}^\ast]$ where no randomness is left and it coincides with the reduced metric value such that~\cite{MitzenmacherUpfal2005}
\begin{align}
	f(\mathbf{c}\odot\mathbf{x}^\ast) \leq \mathbb{E}_\mathbf{X}[f(\mathbf{c}\odot\mathbf{X})].
	\label{eqn:CEguarantee}
\end{align}


\subsection{Distribution of $s(\mathbf{c}\odot \mathbf{Y}^\pm_j)$}

Calculation of the conditional expectations is the main challenge in employing the CE Method.  
For the sake of brevity, let
\begin{equation*}
	\mathbf{Y}^\pm_j\triangleq[x^\ast_0,x^\ast_1,\ldots,x^\ast_{j-1},\pm1,X_{j+1},\ldots,X_{N-1}]^T.
\end{equation*}
The conditional expectations \eqref{eqn:gDefinition} required for the sign decision at iteration $j$ for SRCM, i.e. $\eta_{\scriptscriptstyle N}(.)$, can be written as
\begin{equation}
	g_j^\pm(\mathbf{c})=\frac{1}{LN} \sum_{n=0}^{LN-1} \mathbb{E} [|s(n, \mathbf{c}\!\odot\! \mathbf{Y}^\pm_j)|^6],
	\label{eqn:CMexpCompact}
\end{equation}
which motivates the study of the distribution of $s(n, \mathbf{c}\odot \mathbf{Y}^\pm_j)$.

Recall that the proposed algorithm performs sign selections for a given $\mathbf{c}\in\mathcal{M}'$ by introducing random signs $\mathbf{X}$. That is, subcarriers are modulated by $c_kX_k, k=0,\ldots,N-1$. Consequently the real and imaginary parts for each summand in \eqref{eq:ofdmSymbolDef} are dependent. Consider the centered random variables
\begin{align*}
	\hat{s}_r(n, \mathbf{c}\!\odot\! \mathbf{Y}^\pm_j) &\! =\! s_r(n, \mathbf{c}\!\odot\! \mathbf{Y}^\pm_j) - \mathbb{E}[s_r(n, \mathbf{c}\!\odot\! \mathbf{Y}^\pm_j)], \nonumber \\
	\hat{s}_i(n, \mathbf{c}\!\odot\! \mathbf{Y}^\pm_j) &\! =\! s_i(n, \mathbf{c}\!\odot\! \mathbf{Y}^\pm_j) - \mathbb{E}[s_i(n, \mathbf{c}\!\odot\! \mathbf{Y}^\pm_j)].
\end{align*}
The following lemma determines the covariances of the signal components in the limit, which is a necessary step in proving their joint Gaussian distribution in Theorem~\ref{thm:normality}.
\begin{lemma}
	\label{lem:variances}
	Consider $j=\alpha N$ where $\alpha\leqslant 1$ is a  rational number. For $\mathbf{C}$ uniformly distributed in $\mathcal{M}'^N$, it holds for $n=0,1,\ldots,LN-1$ that
	\begin{align}
		\lim_{N\to\infty} \mathbb{E}_{\mathbf{Y}^\pm_j}[\hat{s}_r(n, \mathbf{C}\!\odot\! \mathbf{Y}^\pm_j) \hat{s}_r(n, \mathbf{C}\!\odot\! \mathbf{Y}^\pm_j)] &= \frac{1-\alpha}{2} \nonumber \\
		\lim_{N\to\infty} \mathbb{E}_{\mathbf{Y}^\pm_j}[\hat{s}_i(n, \mathbf{C}\!\odot\! \mathbf{Y}^\pm_j) \hat{s}_i(n, \mathbf{C}\!\odot\! \mathbf{Y}^\pm_j)] &= \frac{1-\alpha}{2} \nonumber \\
		\lim_{N\to\infty} \mathbb{E}_{\mathbf{Y}^\pm_j}[\hat{s}_r(n, \mathbf{C}\!\odot\! \mathbf{Y}^\pm_j) \hat{s}_i(n, \mathbf{C}\!\odot\! \mathbf{Y}^\pm_j)] &= 0
		\label{eqn:variances}
	\end{align}
with probability one. 
\end{lemma}
\begin{proof}
The proof consists of showing that the variances of the  random expected values in  \eqref{eqn:variances} converge to zero. Consequently, it implies that they are almost surely equal to the limit value of their expected values as $N\to\infty$. If $j$ is constant and not growing with $N$, the limit value is identical to the case where no sign decision is made. Therefore, it is not reflected in the lemma. The case of $j=\alpha N$, which is pertinent to the working of the algorithm, alters the limit value and must be accounted for.  The full proof is rather technical and is omitted due to lack of space.  
\end{proof}

The next essential property of the signal samples is the joint Gaussian distribution of their real and imaginary parts, which is shown in the following theorem. Given Lemma~\ref{lem:variances}, a standard procedure involving Cram\'{e}r-Wold Theorem and Lindeberg's Condition  \cite{billingsley1999} can be used to show the joint distribution. The proof is omitted due to the lack of space.

\begin{theorem}
	\label{thm:normality}
	For any given~$\mathbf{c}\in\mathcal{M}'$ and $n=0,1,\ldots,LN-1$,
	\begin{equation}
		\begin{bmatrix}
			\hat{s}_r(n, \mathbf{c}\!\odot\! \mathbf{Y}^\pm_j) \\
			\hat{s}_i(n, \mathbf{c}\!\odot\! \mathbf{Y}^\pm_j)
		\end{bmatrix} 
		\xrightarrow{d} \mathcal{N}(\mathbf{0},\frac{1-\alpha}{2} I).
		\label{eqn:jointNormalS}
	\end{equation}
	where $\mathbf{0}$ is a vector of zeros and $I$ is a $2\times 2$ identity matrix. 
\end{theorem}

\subsection{Derivation of the Conditional Expectations}
\label{sec:CE-calc}

The expected values in \eqref{eqn:CMexpCompact} are the third moments of $|s(n,\mathbf{c}\odot \mathbf{Y}^\pm_j)|^2$. As shown in Theorem~\ref{thm:normality}, the real and imaginary parts of a signal sample $s(n,\mathbf{c}\odot\mathbf{Y}_j^\pm)$ are independent in the limit due to the zero covariance. Therefore, 
$|s(n,\mathbf{c}\odot \mathbf{Y}^\pm_j)|^2$ can be approximated for large enough $N-j$ as a non-central $\chi^2$-distributed random variable with two degrees of freedom.

Let $\sigma_{s,j}^2=\frac{1}{2}(1-\alpha)$, the value obtained in Lemma~\ref{lem:variances}. Then the real and imaginary parts of
\begin{equation*}
	z(n, \mathbf{c}\!\odot\! \mathbf{Y}^\pm_j)=\sigma_{s,j}^{-1} s(n, \mathbf{c}\!\odot\! \mathbf{Y}^\pm_j)
\end{equation*}
have approximately unit variance for large $N$ with accordingly scaled expected values. Consequently, the third moment can be obtained from the moment generating function of the $\chi^2$ random variable $|z(n, \mathbf{c}\odot \mathbf{Y}^\pm_j)|^2$, for which closed-form expressions are available. That is,
\begin{equation}
	\mathbb{E} [|s(n, \mathbf{c}\!\odot\! \mathbf{Y}^\pm_j)|^6]= \sigma_{s,j}^6 \ \frac{d^3 M^\pm_{j,n}(t)}{dt^3}\Big|_{t=0},
	\label{eqn:expSinZ}
\end{equation}
\saeedDissertation{ The right hand side expectation can be obtained by
\begin{align}
	\mathbb{E} [|z(n, \mathbf{c}\!\odot\! \mathbf{Y}^\pm_j)|^6]= \frac{d^3 M^\pm_{j,n}(t)}{dt^3}\Big|_{t=0},
	\label{eqn:Mderivative}
\end{align}}
with the moment generating function
\begin{align*}
	M^\pm_{j,n}(t)=e^{\lambda^\pm_{j,n} t (1-2t)^{-1}} (1-2t)^{-1} \quad 2t<1,
\end{align*}
where  the non-centrality parameter is
\begin{equation*}
	\lambda^\pm_{j,n}=\sigma_{s,j}^{-2 }\left| \mathbb{E}\left[s(n,\mathbf{c}\!\odot\! \mathbf{Y}^\pm_j)\right]\right| ^2,
\end{equation*}
where the expected values can be constructed cumulatively by adding the contribution of one subcarrier at each iteration.

Obtaining the derivative in \eqref{eqn:expSinZ} and substituting it in \eqref{eqn:CMexpCompact}, we have
\begin{align*}
	g_j^\pm(\mathbf{c})(\mathbf{c})\!= \!\frac{\sigma_{s,j}^6}{LN}  \!\sum_{n=0}^{LN-1}\!\! \left[(\lambda^\pm_{j,n})^3 \!+ \!18(\lambda^\pm_{j,n})^2 + 72\lambda^\pm_{j,n}+ \!48\right]
\end{align*}
which leads to the following closed-form expression for the decision rule \eqref{eqn:CEGeneralRule} for sign selection.

\textbf{Decision rule} \textit{ In the $j$th iteration of the algorithm, the value of $x_j^\ast$ can be decided by
\begin{align}
	x_j^\ast&= -\mathrm{sign}\bigg(\sum_{n=0}^{LN-1} \Big[(\lambda^+_{j,n})^3 + 18(\lambda^+_{j,n})^2 + 72\lambda^+_{j,n}-   \nonumber \\
	&  \quad \quad  \quad \quad  (\lambda^-_{j,n})^3 - 18(\lambda^-_{j,n})^2 - 72\lambda^-_{j,n}\Big] \bigg),
	\label{eqn:finalRule}
\end{align}
where the non-centrality parameters depend on $\mathbf{c}$ and $x^\ast_{0:j-1}$.}

\textbf{Remark} The derivations in this section rely on the Gaussian approximation of $s(n, \mathbf{c}\odot \mathbf{Y}^\pm_j)$ at each iteration of the algorithm, whose accuracy depends on the number remaining random sign variables, i.e. $N-j$. That is, the decision rule in~\eqref{eqn:finalRule} does not implement the CE Method for the last few iterations in principle. Using sample average to calculate the conditional expectations for these last iterations is a trivial solution, which can be done accurately using a high number of realizations of the random sign variables. However, simulations have shown that for varying number of last iterations, the CE Method using an accurate sample average delivers the same performance as using exactly the decision rule in~\eqref{eqn:finalRule}. Therefore, \eqref{eqn:finalRule} is advised for all sign decisions in this work. The method is summarized in Algorithm~\ref{alg:ce-cm}.

\begin{algorithm}
\caption{SRCM reduction by the CE Method.}
	\label{alg:ce-cm}
	\begin{algorithmic}[1]
		\REQUIRE $c_0, \ldots, c_{N_f-1}$: data symbols in $\mathcal{M}$\\  $c_{N_f}, \ldots, c_{N-1}$: data symbols in $\mathcal{M}'$
		\\\COMMENT{Element-wise operations on vectors are assumed in the following calculations.}
		\STATE $N_\circ \gets LN$
		\STATE $\mathbf{x}^\ast \gets [1, 1, \ldots, 1]_{N\times 1}$
		\STATE $\mathbf{n} \gets [0, 1, 2,\ldots, N_\circ-1]$
		\STATE $\mathbf{h} \gets \sum_{j=0}^{N_f-1} c_j \exp(2\pi j \mathbf{n}/N_\circ)$ 
		\FOR {$j = N_f-1$ to $N-1$}
			\STATE $\mathbf{p} \gets \mathbf{h}+  c_j \exp(2\pi j \mathbf{n}/N_\circ)$
			\STATE $\mathbf{m} \gets \mathbf{h}-  c_j \exp(2\pi j \mathbf{n}/N_\circ)$
			\STATE $x_j^\ast \gets -\mathrm{sign}(\mathrm{sum}(|\mathbf{p}|^6\! + 18|\mathbf{p}|^4 +72|\mathbf{p}|^2 -|\mathbf{m}|^6\! - 18|\mathbf{m}|^4\! - 72|\mathbf{m}|^2))$
			\IF{$x_j^\ast=1$}
				\STATE $\mathbf{h} \gets \mathbf{p}$
			\ELSE
				\STATE $\mathbf{h}\gets \mathbf{m}$
			\ENDIF
		\ENDFOR
		\RETURN $\mathbf{x}^\ast$
	\end{algorithmic}
\end{algorithm}

\textbf{Remark} It was observed via simulations that the trajectory of the conditional expectations, as the algorithm performs sign decisions for $x^\ast_0$ to $x^\ast_{N-1}$, suggests that the reduction steps become statistically larger. This motivates pruning the algorithm. That is, a solution to the minimization problem stated in \eqref{eqn:optimizationProblem} with fewer reserved sign bits, i.e. lower rate loss, can be obtained by using the first $N_f$ sign variables to bear data and choosing the remaining sign variables $x_{N_f+1:N-1}$ according to \eqref{eqn:CEGeneralRule}.

\section{Performance analysis}
\label{sec:analysis}

The following theorem shows an upperbound on the reduced SRCM for all data vectors $\mathbf{C}$, i.e. including the worst-case reduced value. 

\begin{theorem}
	\label{thm:upperbound}
	For all $\mathbf{C}\in\mathcal{M}'^N$, the reduced SRCM value in the limit is bounded from above as
	\begin{align}
		\lim_{N\to\infty}\eta_{\scriptscriptstyle N}(\mathbf{C}\odot\mathbf{x}^\ast) \leq 6,
		\label{eq:SRCMupper}
	\end{align}
	where $\mathbf{x}^\ast$ is the solution provided by \eqref{eqn:finalRule} for $\mathbf{C}$.
\end{theorem}

 \begin{figure}[t]
	\centering
	\includegraphics[width=\columnwidth]{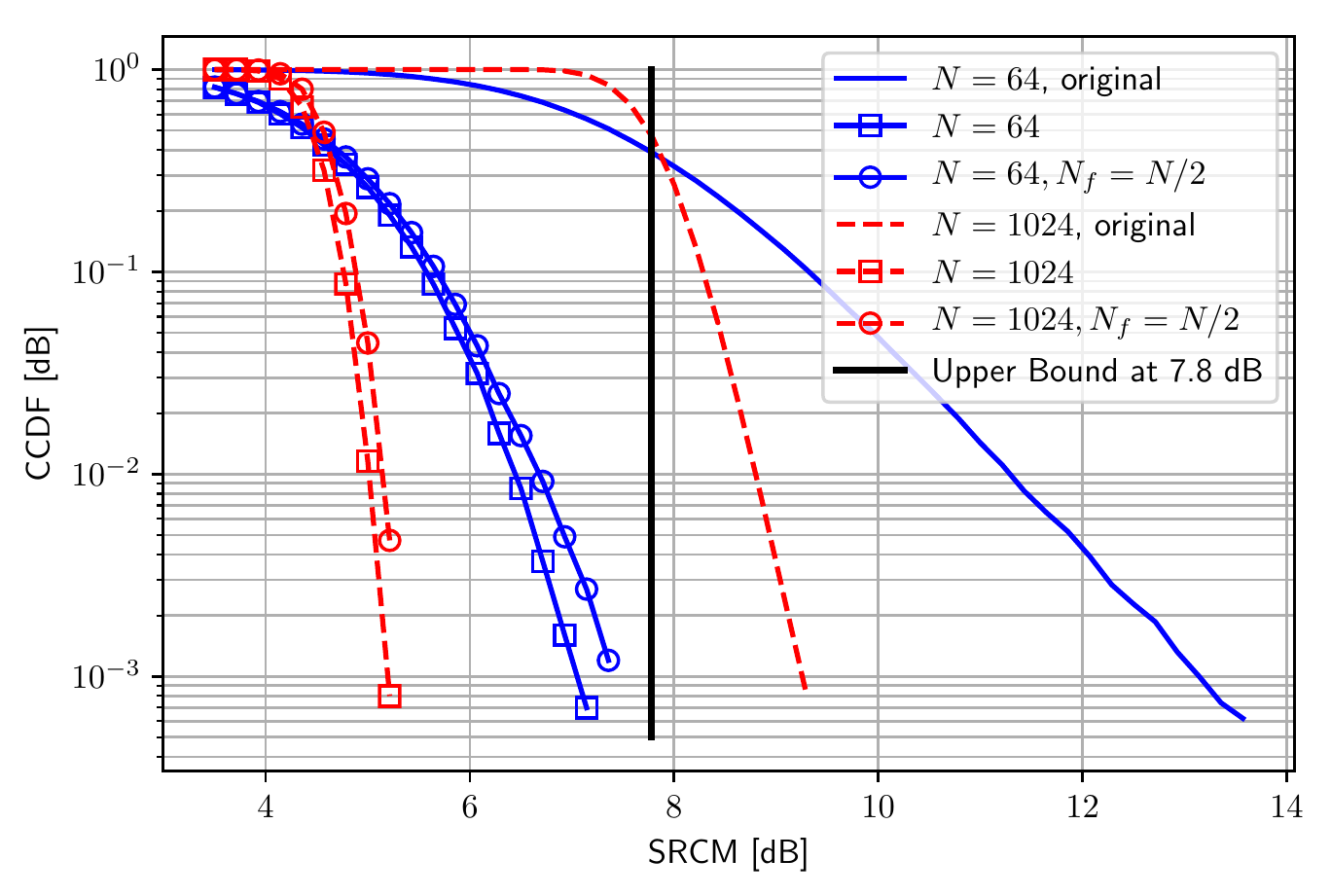}
	\caption{SRCM Reduction performance, including the analytical upperbound.}
	\label{fig:CMreduction}
	\vspace{-0.0cm}
\end{figure}

\begin{proof}
	\saeedLater{we don't need this, we can build the result on the fixed c case in the same manner, reyleigh and ...} As stated in~\eqref{eqn:CEguarantee}, the CE Method guarantees that
	\begin{equation}
		\eta_{\scriptscriptstyle N}(\mathbf{C}\odot\mathbf{x}^\ast) \leq \mathbb{E}_\mathbf{X} [\eta_{\scriptscriptstyle N}(\mathbf{C}\odot\mathbf{X})], \quad \text{for any } N. 
		\label{eq:CEguaranteeCM}
	\end{equation}
In addition, from Theorem~\ref{thm:normality}, we have
	\begin{align*}
		\begin{bmatrix}
			s_r(n, \mathbf{c}\!\odot\! \mathbf{X}) \\
			s_i(n, \mathbf{c}\!\odot\! \mathbf{X})
		\end{bmatrix} 
		\xrightarrow{d} \mathcal{N}(\mathbf{0},\frac{1}{2} I).
	\end{align*}
That is, the distribution of $s(n,\mathbf{c}\odot \mathbf{X})$ in the limit is the same as that of $s(n,\mathbf{B})$, where $\mathbf{B}$  is randomly distributed in $\mathcal{M}$. Therefore, the right hand side of~\eqref{eq:CEguaranteeCM} as $N\to\infty$ is
	\begin{equation}
		\lim_{N\to\infty} \mathbb{E}_\mathbf{X} [\eta_{\scriptscriptstyle N}(\mathbf{C}\odot\mathbf{X})] = \lim_{N\to\infty} \mathbb{E} [\eta_{\scriptscriptstyle N}(\mathbf{B})].
		\label{eqn:Elim}
	\end{equation}
	The distribution of $\eta_{\scriptscriptstyle N}(\mathbf{B})$ is studied in \cite{KimCubicMetric2016}, where it is viewed as the sample average of $|s(n,\mathbf{B})|^6$  with a vanishing variance as $N$ grows. A power of a Rayleigh random variable, here $|s(n,\mathbf{B})|$ as $N\to\infty$, has Weibull distribution with known closed-form expression for its expected value. Consequently, it can be shown that \cite{KimCubicMetric2016}
\begin{equation}
	\lim_{N\to\infty}\mathbb{E}[\eta_{\scriptscriptstyle N}(\mathbf{B})] = 6.
	\label{eqn:RCMlimit}
\end{equation}
Considering \eqref{eq:CEguaranteeCM}, \eqref{eqn:Elim} and that an inequality between two sequences is preserved in their limits, we have
\begin{align*}
		\lim_{N\to\infty}\eta_{\scriptscriptstyle N}(\mathbf{C}\odot\mathbf{x}^\ast) \leq \lim_{N\to\infty}\mathbb{E}[\eta_{\scriptscriptstyle N}(\mathbf{B})],
	\end{align*}
which completes the proof.
\end{proof}

\section{Simulation and Results}
\label{sec:perf}

The reduction in SRCM is shown in Fig.~\ref{fig:CMreduction} for $N=64$ and 1024 to cover a wide range of subcarrier numbers. The constellation size affects the performance but only slightly. Therefore, the simulation results are depicted only for 16-QAM. As shown in the figure, the performance of the pruned algorithm with $N_f=\frac{N}{2}$, i.e. using the second half of sign bits, is only slightly degraded compared to the $N_f=0$ case. It implies that the full rate loss  of the Sign Selection problem, i.e. $\log_{|M|} 2=0.25$, can be reduced to $\frac{1}{2}\log_{|M|} 2=0.125$ with negligible degradation in the performance. The analytical upperbound of Theorem~\ref{thm:upperbound} is as well included in Fig.~\ref{fig:CMreduction}, which confirms the analysis for large $N$, as well as the approximation of $\mathbb{E} [\eta_{\scriptscriptstyle N}(\mathbf{B})], \mathbf{B}\in\mathcal{M}^N$ by the constant $10\log_{10} 6 = 7.8 \text{ dB}$. 

As shown in Fig.~\ref{fig:CMreduction}, RCM is reduced from 7.7~dB to 4.5~dB for both $N=64$ and 1024. That is, a surprising result of nearly 3.2~dB reduction practically regardless of $N$. For $N=512$, for which $K_\mathrm{slp}$ and $K_\mathrm{bw}$ were reported in \cite{CMmotorola}, the CM is reduced to 2.87~dB. The available values are presented in Table~\ref{tbl:CMperf}. The reduction performance is compared to the well-known method Selected Mapping (SLM) \cite{543811}. The result is shown in Fig.~\ref{fig:CMreduction2} for the relatively large number of candidate representations $S=100$, where it can be seen that the performance of the proposed algorithm is significantly better than SLM for $N=1024$, but almost the same for $N=64$. That is, the SLM method cannot maintain its performance as $N$ increases. 

\begin{figure}[t]
	\centering
	\includegraphics[width=\columnwidth]{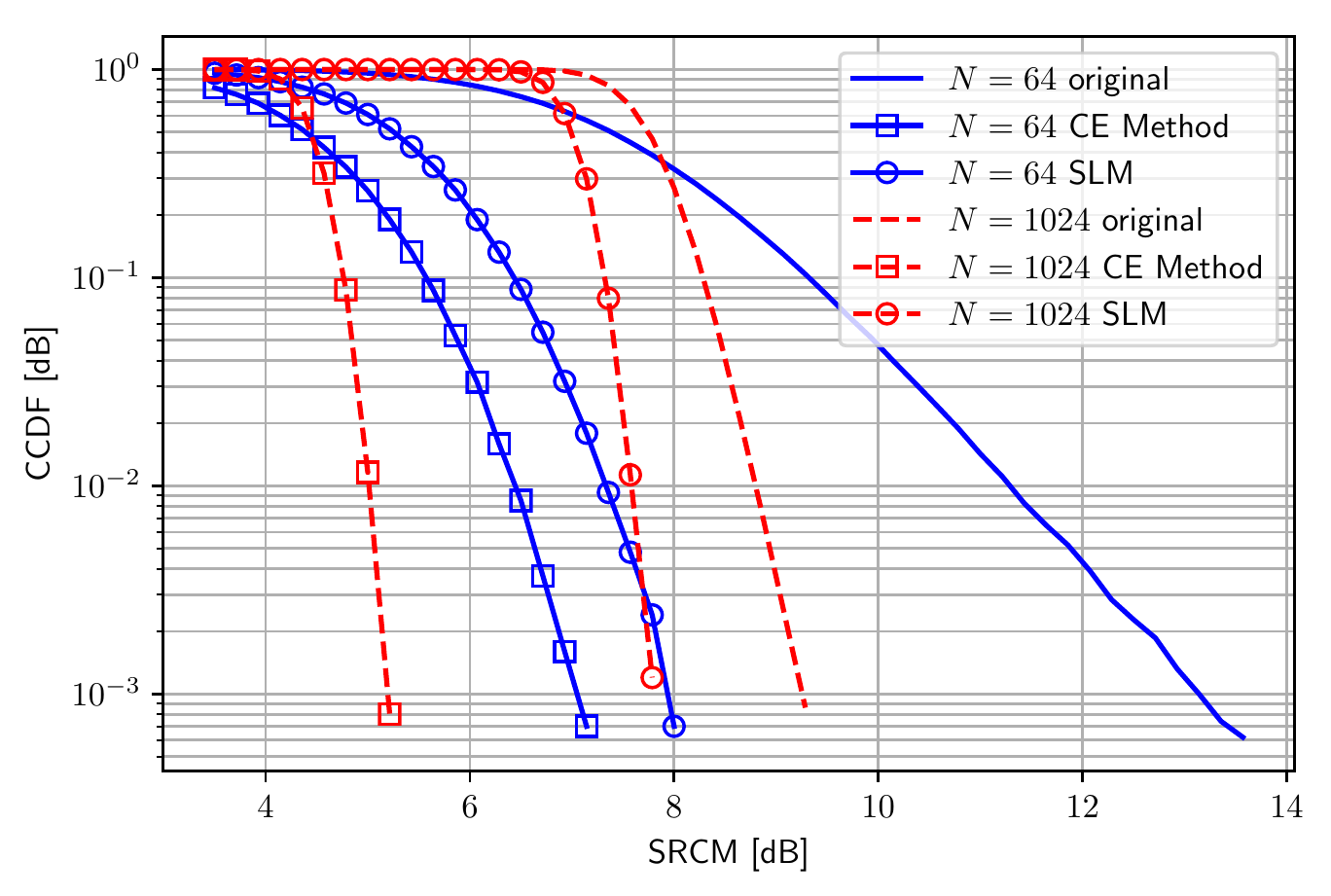}
	\caption{Comparison of the CE Method with SLM with $S=100$.}
	\label{fig:CMreduction2}
	\vspace{-0.4cm}
\end{figure}

For each iteration, the expected value of the signal is constructed cumulatively by adding contribution of each subcarrier at the end of each iteration, which requires $LN$ complex multiplications. Let $\mathrm{D}_1$ be the constant number of complex multiplications required for calculations on each signal sample. Consequently, we have $\mathrm{D}=L(\mathrm{D}_1+1) N^2$ complex multiplications required per OFDM symbol. This is indeed a tentative discussion on the complexity as calculating the constant $\mathrm{D}_1$ is a hardware-related matter.

\section{Conclusion}

The CE Method was used to obtain a suboptimal solution to the Sign Selection problem for SRCM, and eventually CM, reduction. The study of  the OFDM signal samples  under the specific changes imposed by the Sign Selection problem showed that their distribution depends on the given data vector in a tractable manner. Thanks to this observation, the CE Method made it possible to exploit the mathematical structure of SRCM to derive a closed-form decision rule for the sign selections. In addition, an accessible worst-case upperbound on the reduced SRCM was derived. Regarding the reduction performance, two remarkable characteristics were observed from the simulation results. Firstly, the reduced RCM is almost constant for a wide range of $N$, for which the possibility of existence of an analytical justification motivates further research. Secondly, the impact of the individual sign decisions varies by the sign index, such that nearly the same performance is achieved by the using only half of the sign variables. Further work is required to gain insight particularly about the effect of the choice of the subset of the sign variables and the order of decisions.

\begin{table}[t]
	\centering
	\caption{RCM Reduction Performance.}
	\begin{tabular}{|l|c|c|c|c|c|}
		\hline 
		$N$ & orig. RCM & orig. CM & reduced RCM  & reduced CM \\
		\hline
		64 & 7.7 dB& -& 4.5 dB& - \\
		512 & 7.8 dB& 4.8 dB&4.5 dB & 2.87 dB\\
		1024 & 7.8 dB& - & 4.5 dB& -\\
		\hline
	\end{tabular}
	\label{tbl:CMperf}
	\vspace{-0.4cm}
\end{table}

\section*{Acknowledgment}
This work was supported by the German Research Foundation (DFG) under the grant WU 598/3-1.

\bibliographystyle{IEEEtran}
\bibliography{paper}

\end{document}